\documentclass[letterpaper, 10pt, conference]{ieeeconf}
\IEEEoverridecommandlockouts
\usepackage{geometry}
\usepackage{graphicx}
\usepackage{epsfig}
\usepackage{amsmath}
\usepackage{amssymb} 
\usepackage[noadjust]{cite}
\usepackage{bm}
\usepackage{subfigure}
\usepackage{acronym}
\usepackage{paralist}
\usepackage{float}
\usepackage{epstopdf}
\usepackage{algorithm}
\usepackage[noend]{algpseudocode}
\usepackage{mathtools}
\usepackage{multicol} 
\usepackage{accents}
\usepackage{breqn}

\makeatletter
\def\BState{\State\hskip-\ALG@thistlm}
\makeatother

\newtheorem{define}{Definition}
\newtheorem{theorem}{Theorem}
\newtheorem{lemma}{Lemma}

\newcommand{\Rc}{\mathcal R}
\newcommand{\R}{\mathbb R}
\newcommand{\Sc}{\mathcal{S}}
\newcommand{\N}{\mathcal{N}}
\newcommand{\D}{\mathcal{D}}
\newcommand{\V}{\mathcal{V}}

\newcommand{\E}{\mathcal{E}}

\newcommand{\Le}{\mathcal{L}}
\newcommand{\A}{\mathcal{A}}

\newcommand{\eps}{\epsilon}

\newcommand{\Z}{\mathbb{Z}}

\newcommand{\bmx}[1]{\begin{bmatrix}#1\end{bmatrix}} 
\newcommand{\nrm}[1]{\left \lVert#1\right \rVert} 

\DeclarePairedDelimiter{\ceil}{\lceil}{\rceil}
\DeclarePairedDelimiter{\floor}{\lfloor}{\rfloor}
\DeclarePairedDelimiter{\abs}{\lvert}{\rvert}

\newcommand{\rarr}{\rightarrow} 

\makeatletter
\let\oldceil\ceil
\def\ceil{\@ifstar{\oldceil}{\oldceil*}}
\makeatother
\makeatletter
\let\oldfloor\floor
\def\floor{\@ifstar{\oldfloor}{\oldfloor*}}
\makeatother
\makeatletter
\let\oldnorm\norm
\def\norm{\@ifstar{\oldnorm}{\oldnorm*}}
\makeatother
\makeatletter
\let\oldabs\abs
\def\abs{\@ifstar{\oldabs}{\oldabs*}}
\makeatother

\newtheorem{Exp conv}[lemma]{Lemma}
\newtheorem{Conv L2}[lemma]{Lemma}
\newtheorem{FT Conv L2}[lemma]{Lemma}
\newtheorem{FT Conv L1}[lemma]{Lemma}

\newtheorem{wij}[theorem]{Theorem}
\newtheorem{Conv L1}[theorem]{The
orem}
\newtheorem{Conv Lp}[theorem]{Theorem}
\newtheorem{FT Conv Lp}[theorem]{Theorem}

\begin{document}

\title{\LARGE \bf  Finite-Time Resilient Formation Control with Bounded Inputs}


\author{James Usevitch, Kunal Garg, and Dimitra Panagou
\thanks{The authors are with the Department of Aerospace Engineering, University of Michigan, Ann Arbor; \texttt{usevitch@umich.edu}, \texttt{kgarg@umich.edu}, \texttt{dpanagou@umich.edu}.}
\thanks{The authors would like to acknowledge the support of the Air Force Office of Scientific Research under award number FA9550-17-1-0284, and the Automotive Research Center (ARC) in accordance with Cooperative Agreement W56HZV-14-2-0001 U.S. Army TARDEC in Warren, MI, and the Award No W911NF-17-1-0526. This work has been funded by the Center for Unmanned Aircraft Systems (C-UAS), a National Science Foundation Industry/University Cooperative Research Center (I/UCRC) under NSF Award No. 1738714 along with significant contributions from C-UAS industry members. }
}

\maketitle
\thispagestyle{empty}
\pagestyle{empty}

\acrodef{wrt}[w.r.t.]{with respect to}
\acrodef{apf}[APF]{Artificial Potential Fields}
\begin{abstract}
In this paper we consider the problem of a multi-agent system achieving a formation in the presence of misbehaving or adversarial agents. We introduce a novel continuous time resilient controller to guarantee that normally behaving agents can converge to a formation with respect to a set of leaders. The controller employs a norm-based filtering mechanism, and unlike most prior algorithms, also incorporates input bounds. In addition, the controller is shown to guarantee convergence in finite time. A sufficient condition for the controller to guarantee convergence is shown to be a graph theoretical structure which we denote as Resilient Directed Acyclic Graph (RDAG). Further, we employ our filtering mechanism on a discrete time system which is shown to have exponential convergence. Our results are demonstrated through simulations.
\end{abstract}

\IEEEpeerreviewmaketitle

\section{Introduction}
\label{intro}

The study of resilient control in the presence of adversarial agents is a rapidly growing field. An ever-growing amount of cyber attacks has led to increasing attention on algorithms that guarantee safety and security despite the influence of faults and malicious behavior. Controllers that protect against adversarial actions are especially critical in distributed systems where agents may have limited power, computational capabilities, and knowledge of the system as a whole.

The problem of agents achieving formation with respect to a leader or set of leaders has been well-studied in the literature under the assumption that all agents are behaving (see \cite{Oh2015survey} and its references). However, it is well known that the introduction of one faulty or misbehaving agent can disrupt the performance of the entire network. The literature has addressed this problem when agents have simply crashed, have actuator or sensor faults, or have malicious intent \cite{Agmon2006fault,Park2018robust,Defago2006fault}. Much work remains to be done in this area however, especially when the misbehaving agents have malicious intent rather than simply being subject to faults.

A certain group of resilient consensus algorithms based on a filtered-mean or median based approach have gained traction recently in the literature. These algorithms include the W-MSR \cite{LeBlanc_2013_Res}, ARC-P \cite{LeBlanc_2013_Res_Continuous}, SW-MSR \cite{Salda2017}, DP-MSR \cite{Dibaji2017resilient}, LFRE \cite{Mitra2016secure}, and MCA \cite{Zhang2012c} algorithms, which have all been used for resilient consensus. There are a few limitations to these prior results. One limitation is that no upper bound is assumed on agents' maximum inputs. Many systems with agents coming to consensus on physical states are modeled as having such a bound. In addition, the ARC-P algorithm has been shown to have asymptotic convergence, but to the best of authors' knowledge a precise convergence rate has not been proven. 

Finite-time consensus has been a popular field of research recently \cite{Liu2016nonsmooth, Wang2016fault, Tong2018finite, Hong2018finite,zheng2012finite} and little prior work has addressed this topic in the case of resilient controllers. The consensus algorithm in \cite{Franceschelli2017} does consider a resilient algorithm with finite time convergence and bounded inputs. However, it considers undirected graphs where all misbehaving agents must be connected only to agents which are guaranteed to be cooperative. The analysis in this paper considers directed graphs and assumes a different adversary model, where all agents are vulnerable to attacks and that the set of adversaries is $F$-local \cite{Leblanc2012consensus}. As we will show, our method does not require the knowledge of a set of agents invulnerable to misbehavior.

Our contributions are as follows: (i) We introduce a novel continuous finite-time controller that allows agents to achieve formations in the presence of adversarial agents. The controller employs a novel filtering mechanism based on the norm of the difference between agents' states. (ii) We prove that this controller guarantees convergence with bounded inputs. (iii) We define novel conditions for the filtering timing and input weights which ensure that agents can remain in formation even with a dwell time in the filtering mechanism. (iv) We show that the norm-based filtering and bounded input elements of our continuous-time controller can be used in a similar resilient discrete-time system, which is proven to have exponential convergence.

Our paper is outlined as follows: in Section \ref{sec:not} we outline our notation and give the problem formulation; in Section \ref{sec:cont sys} and \ref{sec:disc sys} we present our main results on resiliently achieving formation in continuous and discrete time, respectively; in Section \ref{sec:sim} we present simulations demonstrating our results; and our conclusions and thoughts on future work are summarized in Section \ref{sec:conc}.

\section{Modeling and Problem Formulation}\label{sec:not}
\subsection{Notation}
We denote a directed graph (digraph) as $\mathcal{D} = (\mathcal{V},\E)$, with $\V = \{1,...,n\}$ denoting the vertex set, or agent set, of the graph and $\E$ denoting the edge set of the graph. A directed edge is denoted as ${(j,i) \in \E}: i,j \in \V$, which implies that $i$ is able to sense or receive information from agent $j$. Note that $(i,j) \neq (j,i)$. We say that agent $j$ is an in-neighbor of $i$ and $i$ is an out-neighbor of $j$. The set of in-neighbors of agent $i$ is denoted $\V_i = \{j: (j,i) \in \E)\}$. Three subsets of $\V$ are considered in this paper: leader agents $\Le$, adversarial agents $\A$, and normal agents $\N$. These subsets will be described in more detail in section \ref{sec:probdef}. We denote $\A_i = \V_i \bigcap \A$, i.e. the set of adversarial agents in the in-neighbour set of agent $i$. The in-neighbors which are \emph{not} filtered out are denoted $\Rc_i \subseteq \V_i$. For brevity of notation, we will denote $\Rc_i^\N = \Rc_i\setminus (\A_i\bigcap\Rc_i)$ and $\Rc_i^\A = \A_i\bigcap\Rc_i$, which implies $\Rc_i = \Rc_i^\N\bigcup \Rc_i^\A$. We also denote $R_i \triangleq |\Rc_i(t)|$. The cardinality of a set $S$ is written as $|S|$, the set of integers as $\mathbb{Z}$, the set of integers greater than or equal to 0 as $\mathbb{Z}_{\geq 0}$, the natural numbers as $\mathbb{N}$ and the set of non-negative reals as $\mathbb R_+$. Finally, $\nrm{\cdot}$ denotes any $p$-norm defined on $\R^n$. The control protocols in this paper involve a process in which each agent $i \in \N$ filters out a subset of its in-neighbors $\V_i$. The details are given in sections \ref{subsec: cont filt} and \ref{subsec: disc filt}.

\subsection{Problem Definition}
\label{sec:probdef}

Consider a time-invariant digraph $\D = (\V,\E)$ of $n$ agents with states $\bm p_i \in \R^n$. Each agent $i \in \V$ has the system model
\begin{alignat}{2}
    \dot{\bm p}_i(t) &= \bm u_i(t)& \text{(Continuous Time)} \\
    \bm p_i[t+1] &= \bm p_i[t] + \bm u_i[t] &\hspace{2em} \text{(Discrete Time)} 
\end{alignat}
where $\bm u_i[t] , \bm u_i(t) \in \R^n$ are the discrete and continuous inputs to agent $i$, which will be explained in sections \ref{sec:disc sys} and \ref{sec:cont sys} respectively.
We assume that there exists a subset of the agents $\A \subset \V$ that is adversarial. These agents apply arbitrary or malicious inputs at each time $t$ and for each $k\in \A$:
\begin{alignat}{2}
    \dot{\bm p}_k &= \bm f_{k,m}(t)\ &\hspace{2em} \text{(Continuous Time)} \label{eq:mal cont}\\
    \bm p_k[t+1] &= \bm p_k[t] + \bm f_{k,m}[t]\  &\hspace{2em} \text{(Discrete Time)} \label{eq:mal disc}
\end{alignat}
Similar to prior literature, we assume that $\A$ is an \emph{$F$-local} set, meaning that for any $i \in (\V \backslash \A),\ |\V_i \bigcap \A| \leq F$. 

There is much prior literature on formation control problems involving a set of leaders to which the rest of the network converges. We assume that a subset of the agents $\Le \subset \V$ are designated to behave as leaders. However, these leaders are not invulnerable to attacks, implying $(\Le \bigcap \A)$ may possibly be nonempty. Any nodes which are neither leaders nor adversarial are designated as normal nodes $\N \subset \V$. In all, $\N \bigcup \Le \bigcup \A = \V$.

We assume that prescribed constant formation vectors $\bm \xi_i \in \R^n$ have been specified for these agents. Each $\bm \xi_i \in \R^n$ represents agent $i$'s desired formational offset from a group reference point. The formation offsets of the entire network is written as $\bm \xi = \bmx{ \bm \xi_1^T & \ldots & \bm \xi_n^T}^T$. As outlined in chapter 6 of \cite{Mesbahi2010}, we define the variable $\bm \tau_i$ as $\bm \tau_i = \bm p_i - \bm \xi_i$
If non-adversarial agents come to formation on their values of $\bm \tau_i$, i.e. $\nrm{\bm \tau_i - \bm \tau_j} \rarr 0\ \forall i,j \in (\Le \bigcup \N) \backslash \A$
then they have achieved formation. The behaving leaders are assumed to be maintaining their $\bm \tau$ values at some arbitrary point $\bm \tau_L$. This is the goal of this paper: to design a control protocol so that all the normal behaving agents can come to formation at $\bm \tau_L$.  We assume that each agent $i$ is able to obtain the vectors $\bm \tau_j - \bm \tau_i$ for all $j \in \V_i$, i.e. each relative vector between $\bm \tau_i$ and $\bm \tau_j$. Agent $i$ either measures this vector via on-board sensors or calculates it by receiving transmitted messages from each $j \in \V_i$. 
Adversarial agents may attempt to prevent their normal out-neighbors from coming to formation by either physically misbehaving in the former case, or by sending false values of its $\bm \tau$ value in the latter. In the latter case they may be Byzantine \cite{LeBlanc_2013_Res_Continuous,Lamport1982} in the sense that they are able to send different $\bm \tau$ values to different out-neighbors at any given time instance or time step. As outlined in \cite{LeBlanc_2013_Res_Continuous}, since in the continuous time case each normal agent will have continuous state trajectories, any discontinuity in an adversarial agent's transmitted signal could expose its misbehavior to the network. Hence we assume that in the continuous time case, any signal $\bm \tau_k(t)$ or $\bm p_k(t)$ received by any normal agent $i \in \N$ from any adversary $k \in \A$ is continuous. The assumption of continuity of $\bm \tau_k(t)$ is also made for the case where agents make on-board measurements.

In this paper we consider two settings: a continuous time (Section \ref{sec:cont sys}) and a discrete-time (Section setting \ref{sec:disc sys}). In each of the case, we describe the control protocol, distance-based filtering algorithm and convergence analysis. 

\subsection{Graph Theoretical Conditions}
Our method employs a graph-theoretical structure which we call a Resilient Directed Acyclic Graph (RDAG). This structure is a special case of a class of graphs called Mode Estimated Directed Acyclic Graphs (MEDAGs) \cite{mitra2018}, and is defined as follows:

\begin{define}
    A digraph $\D = (\V,\E)$ is a \emph{Resilient Directed Acyclic Graph} (RDAG) with parameter $r$ if it satisfies the following properties for an integer $r \geq 0$:
    \begin{enumerate}
        \item There exists a partitioning of $\V$ into sets $\Sc_0, \ldots, \Sc_m \subset \V,\ m \in \Z$ such that $|\Sc_j| \geq r$ for all $0 \leq j \leq m$.
        \item For each $i \in \Sc_j,\ 1 \leq j \leq m$, $\V_i \subseteq \bigcup_{k=0}^{j-1} S_k$ 
    \end{enumerate}
\end{define}

Intuitively, an RDAG is a graph defined by successive subsets of agents $\Sc_j$. Agents in each subset only have in-neighbors from preceding subsets. The purpose of an RDAG is to introduce enough edge redundancy to ensure the existence of an unfiltered directed path of behaving nodes from the leaders to each normal agent. This can be achieved by designating all agents in the set $\Sc_0$ to behave as leaders, i.e. $\Sc_0 = \Le$. In our analysis, we consider RDAGs with parameter $r \geq 3F+1$. As we will show, this condition will guarantee that normal agents applying our controllers will converge to the leaders. The existence of an RDAG graph structure does not guarantee that normal agents are able to identify adversarial agents. Rather, the edge redundancy guarantees that each normal agent has enough behaving in-neighbors to still achieve formation under the proposed controllers. Methods have been introduced by which RDAGs can be constructed from existing graph topologies, even in the presence of adversaries (\cite{Mitra2016secure}). In particular, an RDAG can be constructed from a graph that is strongly robust with respect to a subset $\Sc \subset \V$. An example of such a graph is a $k$-circulant graph \cite{Usevitch2018res}. 

\section{Continuous-time System}\label{sec:cont sys}
\subsection{Filtering Algorithm and Control Law}\label{subsec: cont filt}

In the continuous time setting, each agent applies Algorithm \ref{alg:cont} at every time instance $t$.



\begin{algorithm}
\caption{Continuous-Time Filtering}\label{alg:cont}
\begin{algorithmic}
    \Procedure{UpdateFilteredList}{}
        \State Calculate $\tau_{ij} = \nrm{\bm \tau_j - \bm \tau_i}\ \forall j \in \V_i$ 
        \If{$t = m \eps_d,\ m \in \Z_{\geq 0}$}
            \State Sort $\tau_{ij}$ values such that $\tau_{ij_1} \geq \ldots \geq \tau_{ij_{|\V_i|}}$
            \State $\Rc_i(t) \leftarrow \{j : \tau_{ij} \in \{\tau_{ij_{F+1}}, \ldots, \tau_{ij_{|\V_i|}} \} \}$
        \EndIf
    \EndProcedure
\end{algorithmic}
\end{algorithm}

The dynamics of continuous time $\bm \tau(t)$ are given as:
\begin{align}\label{cont dyn}
    \dot{\bm\tau}_i(t) = \dot{\bm p}_i(t) - \dot{\bm \xi}_i(t) =\bm u_i(t).
\end{align}
We assume that the speed of each agent $i$ is bounded above by $u_M$, i.e. $\|\bm u_i(t)\| \leq u_M$ for all $t\geq 0$. Under this constraint, the saturation function is defined as
\begin{align}
    \sigma_i(t) & = \min\{\|\bm u_i^p(t)\|, u_M\},\\
    \bm u_i^p(t) &= \sum_{j\in\Rc_i(t)} w_{ij}(t) (\bm \tau_j(t)-\bm \tau_i(t))\|\bm \tau_j-\bm \tau_i\|^{\alpha-1},
\end{align}
where $0<\alpha<1$. To simplify the notation, define the term $\gamma_i(t) = \frac{\sigma_i(t)}{\|\bm u_i^p\|}$. With this saturation function\footnote{For all $t\geq 0$, $0 \leq \gamma_i(t)\leq 1$. Note that if the distances of agent from its in-neighbours $j\in \Rc_i$ are finite, then $\gamma_i(t)$ is strictly positive.}, inspired from the control law used in \cite{wang2010finite} and using results from \cite{kgarg2018acc}, we define the continuous time control law as:
\begin{align}\label{cont control}
    \bm u_i(t) = \sum_{j \in \mathcal{R}_i(t)} \gamma_i(t) w_{ij}(t)(\bm \tau_j(t) - \bm \tau_i(t))\|\bm \tau_j(t) -\bm \tau_i(t) \|^{\alpha-1}
\end{align}
where $0<\alpha<1$. It can be readily verified from \eqref{cont control} that $\|\bm u_i(t)\| \leq u_M$ for all $t\geq 0$ and that the control input goes to 0 as agent $i$ goes to its equilibrium \footnote{As $\bm \tau_j\rightarrow\bm \tau_i$, term $(\bm \tau_j - \bm \tau_i)\|\bm \tau_j-\bm \tau_i\|^{\alpha-1} \rightarrow 0$ for $\alpha>0$}. 
Note that for $\alpha = 1$, the control law \eqref{cont control} is same as the traditional formation control law (see \cite{Olfati-Saber2004} for example), while for $\alpha = 0$, we get a control law similar to the one introduced in \cite{chen2011finite}. We make use of this type of a controller to not only ensure that $\bm \tau_i$ converges to $\bm \tau_L$, but does so in finite time.

As opposed to \cite{LeBlanc2017}, this protocol is designed such that agents do not update their filtered list  $\Rc_i(t)$ at every time instance $t$, but instead only at time instances $t_1,t_2,t_3, ...$ while keeping it constant during the interval $(t_l,t_{l+1})$. Each of these intervals have constant length, i.e. $t_{l+1}-t_l = \epsilon_d$ for all $l\in \{1,2,3,...\}$ where $\epsilon_d>0$ is a small, positive constant. 
The weights $w_{ij}(t)$ for all $i \in \N$  are designed such that malicious agents are not able to exploit this behavior of $\Rc_i(t)$. Let $\Omega_i(t)$ be the set of in-neighbour agents whose $\bm \tau$ vectors are NOT equal to that of agent $i$, i.e. $\Omega_i(t) = \{j \in \V_i : \nrm{\bm \tau_j - \bm \tau_i} > 0 \}$. Then for all $i \in \N$, we define the control weights $w_{ij}(t)$ for all $j \in \Rc_i(t)$ as 
\begin{align}\label{wij omega}
    w_{ij} (t)=  \left\{
	\begin{array}{lc}
	 0, & \hbox{$|\Omega_i(t)| 
	\leq F$,}\\
	\frac{1}{R_i}, & \hbox{$|\Omega_i(t)| > F$.} \\
	\end{array}
	\right.
\end{align}

To the authors' best knowledge, this choice of control weights have never been introduced in the prior literature. Intuitively, the scheme implies that each normal agent $i$ will have a velocity of zero if its $\bm \tau$ is co-located with the $\bm \tau$ of all but at most $F$ of its in-neighbors. We impose this constraint to ensure that when all normal agents' $\bm \tau$ values have converged to $\bm \tau_L$, the malicious agents are not able to perturb them away from $\bm \tau_L$ during the dwell time. This could happen, for example, if for some $i \in \N$, $\nrm{\tau_i - \tau_k} = 0$ at all $t = m\eps_d$ and $\nrm{\tau_i - \tau_k} > 0$ for time $ t\in (m\eps_d, (m+1)\eps_d)$, where $k \in \A_i$, $m \in \Z_{\geq 0}$. Since $\Rc_i(t)$ is constant for each $t \in [m\eps_d, (m+1)\eps_d)$, the malicious agents would not be filtered out by agent $i$. The properties we impose on the weights prevent the malicious agents to steer the normal agents away during such period.

\begin{wij}\label{wij}
For each agent $i\in \N$, $|\Omega_i(t)|\leq F$ (or, $w_{ij}(t) = 0 \; \forall \; j\in \Rc_i$) for all $t\geq t_i$, if and only if $\bm \tau_i(t) = \bm \tau_L$ for all $t\geq t_i$, for some time $t_i$.
\end{wij}
\begin{proof}
Sufficiency: Assume that there exists some time instant $t_i$ such that for all future times $t\geq t_i$, ${\|\bm \tau_i(t) - \bm \tau_L\|\equiv 0}$. This can only happen if all the filtered in-neighbors of the agent $i$ (i.e. $j \in \Rc_i)$ are at $\bm \tau_L$. To see why this is true, assume that there exists a filtered in-neighbour of agent $i$ which is not at $\bm \tau_L$. Then, by the virtue of the control law \eqref{cont control}, agent $i$ would have a non-zero control input $\bm u_i(t)$, which is a contradiction to the assumption that agent stays at the point $\bm \tau_L$. Hence, all its filtered in-neighbours are at the point $\bm \tau_L$. Since we assume that there are at most $F$ agents in the filtered set $\V_i \backslash \Rc_i$, we get that at most these $F$ agents may not be at $\tau_L$, i.e. $|\Omega_i(t)| \leq F$ and $w_{ij}(t) = 0 \; \forall j\in \Rc_i$.

Necessity: We prove this by contradiction. Let us assume that there exist $\bm \tau^* \neq \bm \tau_L$ and a time $t_i$ such that $\bm \tau_i(t) = \bm\tau^*$ and in addition we have that $|\Omega_i(t)|\leq F$ for all $t\geq t_i$. Let us assume that $i\in \Sc_p$. Since $|\V_i| \geq 3F+1$ and $|\Omega_i(t)|\leq F$, there are at least $2F+1$ in-neighbors which are also staying at $\bm \tau^*$. This implies that there is at least one normal behaving agent in the in-neighbour set of agent $i$ in the set $\bigcup_{l = 0}^{p-1}\Sc_l$, which stays at $\bm \tau^*$. This in turn means that one of its normal behaving in-neighbors in the set $\bigcup_{l = 0}^{p-2}\Sc_l$ stays identically at $\bm \tau^*$. Using this argument recursively, we get that there exists a normal in-neighbor in the set $\Sc_0$, which stays identically at the location $\bm \tau^*$. Since all the normal behaving in-neighbors $\Sc_0$ stay at $\bm \tau_L$, this contradicts the assumption $\bm\tau^* \neq\bm \tau_L$. Hence, we get $\bm \tau_i^* = \bm \tau_L$, and that $|\Omega(t)| \leq F$ for all $t\geq t_i$ only if $\bm \tau_i(t) = \bm \tau_L$ for all $t\geq t_i$.  
\end{proof}

\subsection{Convergence Analysis}
\label{sec:convanaly}
We now prove that under the control law \eqref{cont control}, filtering Algorithm \ref{alg:cont}, and the definition of control weights $w_{ij}$ in \eqref{wij omega}, all the normal behaving agents achieve formation in finite time, despite the presence of adversarial agents. We omit the argument $t$ for the sake of brevity when the dependence on $t$ is clear from the context. First, it is shown that for each normal agent $i\in \Sc_1$, $\|\bm \tau_i(t)-\bm \tau_L\|$ converges to 0 in finite time:
\begin{FT Conv L1}
Consider a digraph $\D$ which is an RDAG with parameter $3F+1$, where $\Sc_0 = \Le$ and $\A$ is an $F$-local set. For each normal agent $i\in \Sc_1$, $\bm \tau_L$ is a globally finite-time stable equilibrium for the closed-loop dynamics \eqref{cont dyn}-\eqref{wij omega}.
\end{FT Conv L1}
\begin{proof}
Choose the candidate Lyapunov function $V(\bm \tau_i) = \frac{1}{2}\|\bm \tau_i-\bm \tau_L\|^2$.
Note that since $\dot{\bm\tau}_i$ is piece-wise continuous in each interval $(t_l,t_{l+1})$, the trajectory $\bm \tau_i(t)$ is piecewise differentiable in each such interval. Let $\dot{\bm \tau}_i(t_{l+1}^-)$ and $\dot{\bm \tau}_i(t_{l+1}^+)$ denote the value of the vector $\dot{\bm \tau}_i$ just before and after the filtering at time instant $t_{l+1}$, respectively. Now, because the right hand side of \eqref{cont control} is bounded at the beginning of each interval, the upper right Dini derivative is defined for $\bm \tau_i(t)$ everywhere, and takes values as
\begin{align*}
    D^+(V(\bm \tau_i))(t) = \left\{
	\begin{array}{lc}
	\nabla V(\bm \tau_i)\dot{\bm \tau}_i(t), & \hbox{$t_l\leq t<t_{l+1}$,}\\
	\nabla V(\bm \tau_i)\dot{\bm \tau}_i(t_{l+1}^+), & \hbox{$t = t_{l+1}$.} \\
	\end{array}
	\right.,
\end{align*}
For the worst case, assume that there are $F$ adversarial agents and $R_i-F$ leaders in the filtered list $\Rc_i$. This requires that the adversarial agent should satisfy $\|\bm \tau_i-\bm \tau_j\|\leq \|\bm \tau_i-\bm \tau_L\|$ for all $j\in \A_i$ and for all $t\geq 0$, otherwise agent $j$ would be filtered out as per Section \ref{subsec: cont filt}. Using this and taking the upper right Dini-derivative of the candidate Lyapunov function along the closed loop trajectories of \eqref{cont dyn}, we get:
\begin{align*}
    D^+(V(\bm \tau_i) )& = (\bm \tau_i-\bm \tau_L)^T \sum_{j\in\Rc_i^\N} \gamma_i w_{ij}(\bm \tau_j - \bm \tau_i)\|\bm \tau_j-\bm \tau_i\|^{\alpha-1}\\
    & + (\bm \tau_i-\bm \tau_L)^T\sum_{j\in\Rc^\A_i} \gamma_i w_{ij}(\bm \tau_j - \bm \tau_i)\|\bm \tau_j-\bm \tau_i\|^{\alpha-1}\\
    & = \gamma_i\frac{R_i-F}{R_i}(\bm \tau_i-\bm \tau_L)^T(\bm \tau_L - \bm \tau_i)\|\bm \tau_L-\bm \tau_i\|^{\alpha-1} \\
    & + (\bm \tau_i-\bm \tau_L)^T\sum_{j\in\Rc^\A_i} \gamma_i w_{ij}(\bm \tau_j - \bm \tau_i)\|\bm \tau_j-\bm \tau_i\|^{\alpha-1}
\end{align*}



Since $\|\bm \tau_i-\bm \tau_j\|\leq \|\bm \tau_i-\bm \tau_L\|$ for all $j\in \Rc_i^\A$, we have:
\begin{align*}
    D^+(V(\bm \tau_i) )& \leq -\gamma_i\frac{R_i-F}{R_i} \|\bm \tau_i-\bm \tau_L\|^{1+\alpha}\\
    + &\gamma_i\sum_{j\in\Rc^\A_i} w_{ij}\|\bm \tau_i-\bm \tau_L\|\|\bm \tau_j - \bm \tau_i\|\|\bm \tau_j-\bm \tau_i\|^{\alpha-1}\\
    \leq&-\gamma_i\frac{R_i-F}{R_i}\|\bm \tau_i-\bm \tau_L\|^{1+\alpha} \\
    + &\gamma_i\frac{F}{R_i-F}\|\bm \tau_i-\bm \tau_L\|\|\bm \tau_L - \bm \tau_i\|\|\bm \tau_L-\bm \tau_i\|^{\alpha-1}\\
\Rightarrow D^+(V(\bm \tau_i) )&\leq -cV(\bm \tau_i)^\beta,
\end{align*}
where $\beta = \frac{1+\alpha}{2}<1$. Note that $D^+(V(\bm \tau_i))\leq 0$ which means that the Lyapunov candidate $V(\bm \tau_i(t))$ is bounded by $V(\bm \tau_i(0))$. This implies that the agent $i$ remains at a bounded distance from the leaders. Also, if any adversarial agent's state moves further away, by the filtering algorithm, they would be filtered out. Hence, each term in $\bm u_i^p$ remains bounded, which in turn means that $\gamma_i(t)>0$. Define $\gamma_i^* = \min\limits_t\gamma_i(t)$. 
Hence, we get that $c \triangleq  \gamma_i^*\frac{R_i-2F}{R_i}>0$. From the results in \cite{bhat2000finite}, since Dini derivative satisfies
$$D^+(V(\bm \tau_i) )\leq -cV(\bm \tau_i)^\beta,$$
for all $\bm \tau_i\in \mathbb R^2$, we get that $\bm \tau_L$ is finite-time stable, with the bound on the finite time of convergence given as:
\begin{align*}
    T_{1i}\leq \frac{V(\bm \tau_i(0))^{1-\alpha}}{c(1-\alpha)} =  \frac{\|\bm \tau_i(0)-\bm \tau_L\|^{1-\alpha}}{c(1-\alpha)}.
\end{align*}
Now, at $t = T_{1i}$, agent $i$ has its $\bm \tau_i$ co-located with all the normal leaders' $\bm \tau$. This means that there can be at max $F$ agents (i.e. the adversarial leaders) which are not co-located with the agent's $\bm \tau_i$. Hence, we get that $|\Omega_i(t)|\leq F$ for all $t\geq T_{1i}$. Therefore, by Theorem \ref{wij} agent $i$ will stay at $\bm \tau_L$ for all future times.  
\end{proof}
Next we take the case of normal agents $i \in \Sc_2$:
\begin{FT Conv L2}\label{FT conv L2}
Consider a digraph $\D$ which is an RDAG with parameter $3F+1$, where $\Sc_0 = \Le$ and $\A$ is an $F$-local set. Under the closed loop dynamics \eqref{cont dyn}-\eqref{wij omega}, the value $\bm \tau_i(t)$ for each normal agent $i\in \Sc_2$ converges to $\bm \tau_L$ in finite time $T_{2i}$. 
\end{FT Conv L2}
\begin{proof}
For the worst case analysis, assume that all the agents in $\Rc_i(0)$ are from $\Sc_1$ and are located such that $(\bm \tau_j(0)-\bm \tau_i(0))^T (\bm \tau_L-\bm \tau_i(0))<0$ for each $j\in \Rc_i(0)$. This simply means that the agents in $\Rc_i$ at time $t = 0$ are located on one side of the agent while the leaders are on the other side. This is the worst case because this arrangement of in-neighbors would make agent $i$ move away from the leaders, initially. Also, assume that $|\Rc^\A_i| = F$ and $|\Rc^\N_i| = R_i-F$, so that agent $i$ has maximum number of adversarial in-neighbours. Consider the candidate Lyapunov function $V(\bm \tau_i(t)) = \frac{1}{2}\|\bm \tau_i(t)-\bm \tau_L\|^2$. Taking its upper right Dini derivative along the closed-loop trajectories of agent $i$, we get $D^+(V(\bm \tau_i) )= (\bm \tau_i-\bm \tau_L)^T\sum_{j\in\Rc_i}  \gamma_i w_{ij}(\bm \tau_j - \bm \tau_i)\|\bm \tau_j-\bm \tau_i\|^{\alpha-1}.$
Now, from the assumption on the initial locations of agents in $\Rc_i(t)$, we get that $D^+(V(\bm \tau_i(0))) = \gamma_i(0)\sum_{j\in \mathcal{R}_i} w_{ij}(\bm \tau_i-\bm \tau_L)^T(\bm \tau_j - \bm \tau_i)\|\bm \tau_j-\bm \tau_i\|^{\alpha-1}>0$. Also, define $T_1 \triangleq \max\limits_{l\in \Sc_1\bigcap\N}T_{1l}$, i.e. $T_1$ is the maximum time after which each normal agent in $\Sc_1$ would achieve formation and have $\bm \tau_i = \bm \tau_L$. Hence, at time $t = T_1$, we get that
\begin{align*}
     D^+(V(\bm \tau_i))&  = \sum_{j\in\Rc_i^\N} \gamma_i w_{ij}(\bm \tau_i-\bm \tau_L)^T(\bm \tau_j - \bm \tau_i)\|\bm \tau_j-\bm \tau_i\|^{\alpha-1} \\
    & + \sum_{j\in\Rc^\A_i} \gamma_i w_{ij}(\bm \tau_i-\bm \tau_L)^T(\bm \tau_j - \bm \tau_i)\|\bm \tau_j-\bm \tau_i\|^{\alpha-1} \\
    &= \gamma_i\frac{R_i-F}{R_i}(\bm \tau_i-\bm \tau_L)^T(\bm \tau_L - \bm \tau_i)\|\bm \tau_L-\bm \tau_i\|^{\alpha-1} \\
    & + \sum_{j\in\Rc^\A_i} \gamma_i w_{ij}(\bm \tau_i-\bm \tau_L)^T(\bm \tau_j - \bm \tau_i)\|\bm \tau_j-\bm \tau_i\|^{\alpha-1} \\
    & \leq -\gamma_i\frac{R_i-F}{R_i}\|\bm \tau_L-\bm \tau_i\|^{1+\alpha} \\
    & +\gamma_i\sum_{j\in\Rc^\A_i} w_{ij}\|\bm \tau_i-\bm \tau_L\|\|\bm \tau_j-\bm \tau_i\|^{\alpha}
\end{align*}
Now, for all $j\in \Rc^\A_i$, the norm $\|\bm \tau_j(T_1)-\bm \tau_i(T_1)\| \leq \|\bm \tau_k(T_1)-\bm \tau_i(T_1)\|$ for some $k\in \Rc_i^\N$ otherwise, these adversarial agents would be filtered out. Using this and the fact that $\bm \tau_k(T_1) = \bm \tau_L$, we get that for all $t\geq T_1$, 
\begin{align*}
    D^+(V(\bm \tau_i(t))) &\leq -\gamma_i\frac{R_i-F}{R_i}\|\bm \tau_L-\bm \tau_i\|^{1+\alpha} \\
    & + \sum_{j\in\Rc^\A_i} \gamma_i w_{ij}\|\bm \tau_i-\bm \tau_L\|\|\bm \tau_L-\bm \tau_i\|^{\alpha} \\
    & = -\gamma_i\frac{R_i-2F}{R_i}\|\bm \tau_L-\bm \tau_i\|^{1+\alpha}<0.
\end{align*}
Since $D^+(V(\bm \tau_i))(0)>0$ while $D^+(V(\bm \tau_i))(T_1)<0$, and it is bounded above in the interval $(0,T_1)$, 
the increment in the value of $V(\bm \tau_i)$ is  bounded in the interval. Hence, agent $i$ would be at a finite distance away from the leaders at time $T_1$. This also implies that $\bm u_i^p(t)$ is bounded and hence $\gamma_i^* = \min\limits_t\gamma_i(t)>0$. Hence, we get that $D^+(V(\bm \tau_i) )\leq -cV(\bm \tau_i)^\beta$ where $c = \gamma_i^*\frac{R_i-2F}{R_i}>0$ and $\beta = \frac{1+\alpha}{2}<1$. Hence, we get that $\bm \tau_i\rightarrow\bm \tau_L$ in finite time. Let $\bm \tau_i(T_1)$ be the position of agent at time instant $T_1$. Using the bound on finite time of convergence, we get that for $t \geq T_{2i}$, $\bm \tau_i(t) = \tau_L$ where 
\begin{align*}
T_{2i} \leq T_1 + \frac{V(\bm \tau_i(T_1))^{1-\alpha}} {c(1-\alpha)} = T_1 +  \frac{\|\bm \tau_i(T_1)-\bm \tau_L\|^{1-\alpha}}{c(1-\alpha)}
\end{align*}
Since both $T_1$ and $\|\bm \tau_i(T_1)-\bm \tau_L\|$ are finite, $\alpha<1$ and $c>0$ we get that $T_{2i}$ is also finite. Again, after time instant $T_{2i}$, agent $i$ has its $\bm \tau_i$ co-located with all the normal in-neighbors' $\bm \tau$. This means that there can be at max F agents (i.e. the adversarial agents) which are not co-located with the agent's $\bm \tau_i$. Hence, we get that $|\Omega_i(t)|\leq F$ for all $t\geq T_{2i}$. Therefore, Theorem \ref{wij} implies that agent $i$ will stay at $\bm \tau_L$ for all $t\geq T_{1i}$.  
\end{proof}
We have shown that each normal agent $i\in \Sc_2$ will achieve the formation in finite time. Now we present the general case:
\begin{FT Conv Lp}
Consider a digraph $\D$ which is an RDAG with parameter $3F+1$, where $\Sc_0 = \Le$ and $\A$ is an $F$-local set. Under the closed loop dynamics \eqref{cont dyn}-\eqref{wij omega}, $\bm \tau_i$ will converge to $\bm \tau_L$ in finite time for all normal agents $i \in \N$.
\end{FT Conv Lp}
\begin{proof}
We have already shown that all the agents in $\Sc_1$ and $\Sc_2$ will achieve formation in finite time. Consider any agent $i \in S_3$. Since all the in-neighbors of agents in $\Sc_3$ are from $\bigcup_{i = 0}^2\Sc_i$, after a finite time period all the agents in $\V_i \bigcap \N$ will satisfy $\bm \tau_i = \bm \tau_L$. Define $T_2 =\triangleq \max\limits_kT_{2k}$, where $k$ belongs to the set of normal agents in $\Sc_1$. After the time instant $t = T_2$, the Lyapunov candidate $V(\bm \tau_i) = \frac{1}{2}\|\bm \tau_i-\bm \tau_L\|^2$ and its Dini derivative will satisfy the conditions similar to Lemma \ref{FT conv L2}. Hence, we get that all the normal agents in $\Sc_3$ will achieve formation in finite time. This time can be bounded as $T_{3i} \leq T_2 +  \frac{\|\bm \tau_i(T_2)-\bm \tau_L\|^{1-\alpha}}{c(1-\alpha)}$ for each $i\in \Sc_3$. This argument can be used recursively to show that each normal agent in $\bigcup_{l =1}^p\Sc_l$ will achieve formation in finite time. Defining $T_l$ as the maximum time by which all the normal agents in set $\Sc_l$ will achieve the formation, one can establish the following relation:
\begin{align*}
    T_{l+1} & \leq T_l + \max_{i\in \Sc_{l+1}}\frac{\|\bm \tau_i(T_l)-\bm \tau_L\|^{1-\alpha}}{c(1-\alpha)} \quad l\geq 1\\
    T_1 & \leq \max_{i\in \Sc_{1}}\frac{\|\bm \tau_i(0)-\bm \tau_L\|^{1-\alpha}}{c(1-\alpha)}
\end{align*}
and since $T_l$ and $\|\bm \tau_i(T_l)-\bm \tau_L\|$ both are finite, we get that $T_{l+1}$ is a finite number. 
\end{proof}
Hence, we have shown that under the effect of our protocol, each normal agent $i$ would achieve formation in finite time, despite adversarial agents. In the next section, we show that our filtering mechanism can be used for the case of discrete time systems as well.
\section{Discrete-time System}\label{sec:disc sys}
\subsection{Filtering Algorithm and Control Law}\label{subsec: disc filt}
At each time step $t$, each agent $i \in \N$ applies the following algorithm:
\begin{algorithm}\caption{Discrete-Time Filtering}\label{alg:disc}
\begin{algorithmic}
    \Procedure{UpdateFilteredList}{}
        \State Calculate $\tau_{ij} = \nrm{\bm \tau_j - \bm \tau_i}\ \forall j \in \V_i$ 
        \State Sort $\tau_{ij}$ values such that $\tau_{ij_1} \geq \ldots \geq \tau_{ij_{|\V_i|}}$
        \State $\Rc_i[t] \leftarrow \{j : \tau_{ij} \in \{\tau_{ij_{F+1}}, \ldots, \tau_{ij_{|\V_i|}} \} \}$
    \EndProcedure
\end{algorithmic}
\end{algorithm}

The discrete time system dynamics are given as
\begin{align}\label{dis dyn}
    \bm \tau_i[t+1] &= \bm p_i[t+1] - \bm \xi_i =\bm p_i[t] + \bm u_i[t] - \bm \xi_i \nonumber\\
    &= \bm \tau_i[t] + \bm u_i[t]
\end{align}
The input of each agent $i$ is bounded above by $u_M > 0$, i.e. $\|\bm u_i[t]\| \leq u_M$ for all $t\geq 0$. Under this constraint, the saturation function is given as
\begin{align}
    \sigma_i[t] & = \min\{\|\bm u_i^p[t]\|, u_M\},\\
    \bm u_i^p[t] &= \sum_{j\in\Rc_i[t]} w_{ij}[t] (\bm \tau_j[t]-\bm \tau_i[t]).
\end{align}
To simplify the notation, define $\gamma_i[t] = \frac{\sigma_i[t]}{\|\bm u_i^p[t]\|}$. We define the control law $\bm u_i[t]$ as
\begin{align}\label{dis control}
   \bm u_i[t] = \gamma_i[t]\sum_{j\in\Rc_i[t]} w_{ij}[t] (\bm \tau_j[t]-\bm \tau_i[t]),
\end{align}
where for all time steps $t$ and for all $i \in \N$, $w_{ij}[t]>0$ and $\sum_{j\in\Rc_i[t]} w_{ij}[t]  = 1$. For simplicity, we choose $w_{ij}[t] = \frac{1}{R_i}$. We point out that $0<\gamma_i[t]\leq 1$. 
In the following subsection, we prove that under the effect of the control law \eqref{dis control} and Algorithm \ref{alg:disc}, normal behaving agents in the discrete time setting are also guaranteed to achieve formation despite the presence of adversarial agents. 

\subsection{Convergence Analysis}
For our analysis, we need the following result:
\begin{lemma}\label{bk exp}
Let $b[k] = kc^kb[0],\ k \in \Z_{\geq 0}$ be a series where $b[0]>0$ and $0 < c<1$. Then there exist positive constants $\alpha ,\beta$ with $c<\beta<1$ such that $\forall k \in \Z_{\geq 0}$,
\begin{align}\label{bk conv}
    b[k] = kc^kb[0] \leq \alpha\beta^k.
\end{align}
\end{lemma}
\begin{proof}
It can be readily verified that for any $c<\beta<1$ and $\alpha\geq \frac{b[0]}{e\log\frac{\beta}{c}}$, the inequality \eqref{bk conv} holds for all $k\geq 0$. 
\end{proof}
First, consider the normal agents in the set $\Sc_1$:
\begin{lemma}
\label{Conv L1}
Consider a digraph $\D$ which is an RDAG with parameter $3F+1$, where $\Sc_0 = \Le$ and $\A$ is an $F$-local set. For every normal agent $i\in \Sc_1$, $\|\bm \tau_i[t]-\bm \tau_L\|$ converges to 0 exponentially.
\end{lemma}

\begin{proof}
For the worst case, assume there are $F$ adversarial agents. Consider any normal agent $i \in \Sc_1$. Since all of its in-neighbours are from $\Sc_0$, we get that $\V_i \subset \Le$ and for all $k \in \V_i\bigcap\N$, $\bm \tau_k = \bm \tau_L$. By definition of an RDAG, $|\V_i| \geq 3F+1$ which implies $R_i \geq 2F+1$ and $|\Rc_i^\N| \geq F+1$. For the worst case, suppose that $\|\bm \tau_i[t]-\bm \tau_j[t]\| \leq \|\bm \tau_i[t]-\bm \tau_L\|$ $\forall j\in \A_i$ so that none of the adversarial agents are filtered out. This implies that $|\Rc_i^\A| = F$ and $|\Rc_i^\N| = R_i-F$. From the closed loop dynamics, we get: 
\begin{align*}
     \bm \tau_i[t+1]-\bm \tau_L & = \bm \tau_i[t] + \sum_{j\in\Rc_i} \gamma_i w_{ij} (\bm \tau_j[t]-\bm \tau_i[t]) - \bm \tau_L.
\end{align*}
Noting that $\Rc_i\subset \Le$, after some manipulation we obtain:
\begin{align}\label{tau t1}
    \bm \tau_i[t+1]-\bm \tau_L & = (1-\gamma_i\frac{R_i-F}{R_i})(\bm \tau_i[t] -\bm \tau_L)\nonumber\\
    & +  \sum_{j\in\Rc_i^\A} \gamma_i w_{ij} (\bm \tau_j[t]-\bm \tau_i[t]).
\end{align}
Since $\|\bm \tau_i[t]-\bm \tau_j[t]\| \leq \|\bm \tau_i[t]-\bm \tau_L[t]\|$ for all $j\in\Rc_i^\A$, we get $\|\sum_{j\in\Rc_i^\A} w_{ij} (\bm \tau_j[t]-\bm \tau_i[t])\| \leq \frac{F}{|\Rc_i|}\|\bm \tau_i[t]-\bm \tau_L \|$. Hence, we get the bound on $\|\bm \tau_i[t+1]-\bm \tau_L\|$ as
\begin{align}\label{tau norm}
    \|\bm \tau_i[t+1]-\bm \tau_L\| \leq (1-\gamma_i\frac{R_i-2F}{R_i})\|\bm \tau_i[t] -\bm \tau_L\|.
\end{align}
Let $\gamma_i^* = \min_k\gamma_i[k]>0$.
Since $1-\gamma_i\frac{R_i-2F}{R_i}\leq 1-\gamma_i^*\frac{R_i-2F}{R_i}<1$, define $c = 1-\gamma_i^*\frac{R_i-2F}{R_i}$, so that we get $ \|\bm \tau_i[t+1]-\bm \tau_L\| \leq c\|\bm \tau_i[t] -\bm \tau_L\|$, i.e. $\|\bm \tau_i[t]-\bm \tau_L\|$ is an exponentially converging sequence. 
\end{proof}

For $i \in \Sc_p$ where $p\geq 2$, we know that there are at most $F$ adversarial agents in $\Rc_i$. Note that by definition of the network RDAG, all agents in $\Rc_i$ are from $\bigcup_{j = 0}^{p-1}\Sc_j$. For the worst-case analysis, we assume there are $F$ adversarial agents in $\Rc_i$ and all the normal agents in $\Rc_i$ are from $\Sc_{p-1}$. From the closed-loop dynamics of the agent $i$, we get:
\begin{align*}
    \bm \tau_i[t+1]- \bm \tau_L & = \bm \tau_i[t] + \sum_{j\in\Rc_i} \gamma_i w_{ij} (\bm \tau_j[t]-\bm\tau_i[t]) - \bm \tau_L,
\end{align*}
which after some manipulation gives:
\begin{align}\label{tau t2}
   &\bm \tau_i[t+1]-\bm \tau_L = (1-\gamma_i\frac{R_i-F}{R_i})(\bm \tau_i[t]-\bm \tau_L) \nonumber\\
    &+\sum_{j\in\Rc_i^\N} \gamma_i w_{ij} (\bm \tau_j[t] - \bm \tau_L) + \sum_{j\in\Rc_i^\A} \gamma_i w_{ij} (\bm \tau_j[t]-\bm\tau_i[t]).
\end{align}
Using the same logic as in Lemma \ref{Conv L1}, we assume for the worst case that $\forall j\in \A_i$, $\|\bm \tau_i[t]-\bm \tau_j[t]\|\leq \|\bm \tau_i[t]-\bm \tau_k[t]\|$ for some $k\in \Rc_i^\N$.
Using this, the fact that $|\Rc_i^\A| = F$, we get
\begin{align*}
\|\sum_{j\in\Rc_i^\N} w_{ij} (\bm \tau_j[t] - \bm \tau_L)\|& \leq \frac{|\Rc_i|-F}{|\Rc_i|}\|\bm \tau_{k}[t] -\bm \tau_L\|,\\
\|\sum_{j\in\Rc_i^\A} w_{ij} (\bm \tau_j[t]-\bm\tau_i[t])\|&\leq \frac{F}{|\Rc_i|}\|\bm \tau_{k}[t] - \bm \tau_i\|.
\end{align*}
We can bound $\|\bm \tau_{k} - \bm \tau_i\| \leq \|\bm \tau_{k} -\bm \tau_L\| + \|\bm \tau_i -\bm \tau_L\|$ to get:
\begin{align}\label{tau norm L2}
   \|\bm \tau_i[t+1]-\bm \tau_L\| &\leq c\|\bm \tau_i[t]-\bm \tau_L\| + \|\bm \tau_{k}-\bm \tau_L\|,
\end{align}
where $c = 1-\gamma_i^*\frac{R_i-2F}{R_i}<1$ where $\gamma_i^*$ is defined as in Lemma \ref{Conv L1}. Inequality \eqref{tau norm L2} is true for every normal agent in $\Sc_p$ with $p\geq 2$. Using this observation, we next consider the case of agents in set $\Sc_2$ :

\begin{Conv L2}\label{l2 conv}
Consider a digraph $\D$ which is an RDAG with parameter $3F+1$, where $\Sc_0 = \Le$ and $\A$ is an $F$-local set. For every normal agent $i\in \Sc_2$, $\|\bm \tau_i[t]-\bm \tau_L\|$ converges to 0 exponentially. 
\end{Conv L2}
\begin{proof}
Define $a[t] = \|\bm \tau_i[t]-\bm \tau_L\|$, $b_k[t] = \|\bm \tau_{k}[t]-\bm \tau_L\|$ so that \eqref{tau norm L2} can be written as $ a[t+1] \leq ca[t] + b_k[t]$:
\begin{align}\label{ak conv}
   a[t+1]& \leq c^{t+1}a[0] + \sum\limits_{i = 0}^t c^{t-i}b_k[i].
\end{align}
Now, $b_k[i]$ represents the norm $\nrm{\bm \tau_k[i] - \bm \tau_L}$ of a normal agent $k\in \Sc_1$, which can be bounded as $b_k[i] \leq c_k^ib_k[0]$ as per \eqref{tau norm} where $c_k = 1-\gamma_k^*\frac{R_k-2F}{R_k}<1$. For the sake of brevity, let $a_0 = a[0]$, $b_{k0} = b_k[0]$. Using this, we get:
\begin{align*}
    a[t+1]& \leq c^{t+1}a_0 + \sum\limits_{i = 0}^t c^{t-i}c_k^ib_{k0} 
\end{align*}
Define $b^*_0 = \max\limits_{k\in \Rc_i^\N}b_{k0}$, $c^* = \max\limits_{k\in \Rc_i^\N} c_k$, and $\tilde c = \max\{c, c^*\}$, so that 
\begin{align*}
    a[t+1]& \leq \tilde c^{t+1}a_0 + \sum\limits_{i = 0}^t \tilde c^tb^*_{0}  = \tilde c^{t+1}a_0 + (t+1) \tilde c^tb^*_{0}. 
\end{align*}
Using this and Lemma \ref{bk exp}, i.e., $k\tilde c^tb_0^* \leq \alpha\beta^t$, we get that:
\begin{align*}
   a[t+1] \leq \tilde c^t(ca_0 +b^*_0) + t\tilde c^tb^*_0 \leq \tilde c^t(ca_0 +b^*_0) + \alpha\beta^t,
\end{align*}
where $\alpha >0$ and $\tilde c<\beta<1$. Now, since $\tilde c<\beta$, we get:
\begin{align*}
    a[t+1] \leq \tilde c^t(ca_0 +b[0]) + \alpha\beta^t \leq \beta^t(ca_0 +b[0] + \alpha).
\end{align*}
As $\beta<1$, $a_t$ converges to 0 exponentially, i.e., for a normal agent $i\in \Sc_2$, $\|\bm \tau_i[t]-\bm \tau_L\|$ converges to 0 exponentially.
\end{proof}
Note that this result can be interpreted as follows: ${\|\bm \tau_i-\bm \tau_L\|}$ for $i \in \N$ converges to 0 exponentially if $\|\bm \tau_j-\bm \tau_L\|$ converges to 0 exponentially for all its normal in-neighbours $j\in \Rc_i\bigcap\N$. Using this, we can state the following result for all normal behaving agents:


\begin{theorem}
Consider a digraph $\D$ which is an RDAG with parameter $3F+1$, where $\Sc_0 = \Le$ and $\A$ is an $F$-local set. Under the closed loop dynamics \eqref{dis dyn}-\eqref{dis control}, $\|\bm \tau_i[t]-\bm \tau_L\|$ converges to 0 exponentially for all agents $i \in \N$. 
\end{theorem}

\begin{proof}
We have proven this result for agents in $\Sc_1$ and $\Sc_2$ in Lemmas \ref{Conv L1} and \ref{l2 conv}. We now consider any node $i \in \Sc_p$ for arbitrary $p$. Observe that every agent $i\in \Sc_p$ satisfies the equation \eqref{ak conv}, where $a[t]$ represents the norm $\|\bm \tau_i[t]-\bm \tau_L\|$ and $b_j[t] = \|\bm \tau_j[t]-\bm \tau_L\|$, where $j\in \bigcup_{l = 0}^{p-1}\Sc_l$. From Lemma \ref{l2 conv}, we have that $\|\bm \tau_i-\bm\tau_L\|$ for normal agents in $\Sc_2$ converges exponentially to 0. Hence, it follows that for each normal agent in $\Sc_3$, $\|\bm \tau_i-\bm \tau_L\|$ converges to 0 exponentially since all of its normal behaving agents are from the set $\bigcup_{l = 0}^2\Sc_l$. Repeating this logic shows that for each normal agent $i \in \Sc_p$, $\|\bm \tau_i-\bm\tau_L\|$ converges exponentially to 0 for each $p\geq 1$. 
\end{proof}

\section{Simulation}\label{sec:sim}
We consider an RDAG of 80 agents with parameter $r = 16$ and $F = 5$. There are 5 sub-levels, $\Sc_l$ with $|\Sc_l| = 16$ for $ l\in \{0,1,2,3,4\}$. The set $\Sc_0$ is composed entirely of agents designated to behave as leaders. In the simulation, $\A$ is a 5-local model with 5 agents in each of the levels $\Sc_l$ becoming adversarial (including in $\Sc_0$). The simulation treats a worst-case scenario in the sense that each agent $i \in \Sc_l,\ l \geq 1$, has in-neighbours only in $\Sc_{l-1}$ and no leader in-neighbors. 
The agents have states in $\R^2$. The vector $\bm \xi$ specifies the formation as points on circle of radius $10 \; \textrm{m}$ centered at $\begin{bmatrix}0 & 10\end{bmatrix}^T$. The vectors $\bm p_i(0),\ i \in \Le$ are chosen such that $\bm \tau_i(0)$ is at the origin for all $i \in \Le$. The vectors $\bm p_j(0)$ for all other agents $j \in (\V_i \backslash \Le)$ are initialized such that their $\bm \tau_j(0)$ values are randomly initialized values. We choose the maximum allowed speed of the agents as $u_M = 1$. These conditions are used for both the continuous and discrete time simulations. For the continuous time case, the control parameter $\alpha$ is chosen as $\alpha = 0.8$. 

Figure \ref{fig:norm cont} shows a plot of $\|\bm \tau_i(t)-\bm \tau_L\|$ versus time for a subset of the normal agents. It is clear that all the normal agents converge to the point where their $\bm \tau$ values are same as those of leader in finite time. Figure \ref{fig:path cont} shows the path $\bm p_i(t) =\bmx{x_i(t) & y_i(t)}^T $ of all the agents and a subset of the adversarial agents. For sake of clarity, only 4 of the 25 misbehaving agents are depicted in the figure. In Figure \ref{fig:path cont} and Figure \ref{fig:path disc}, it can be noted that while some normal agents (belonging to set $\Sc_1$ move directly towards their desired locations, other normal agents first move away from their desired locations. This is in agreement with our analysis; malicious agents are able to exert a bounded influence on normal agents in $\Sc_l$, $l \geq 2$ which do not have any leaders as in-neighbors, while convergence is still guaranteed in a finite time period.

\begin{figure}[!htbp]
	\centering
	\includegraphics[width=0.9\columnwidth,clip]{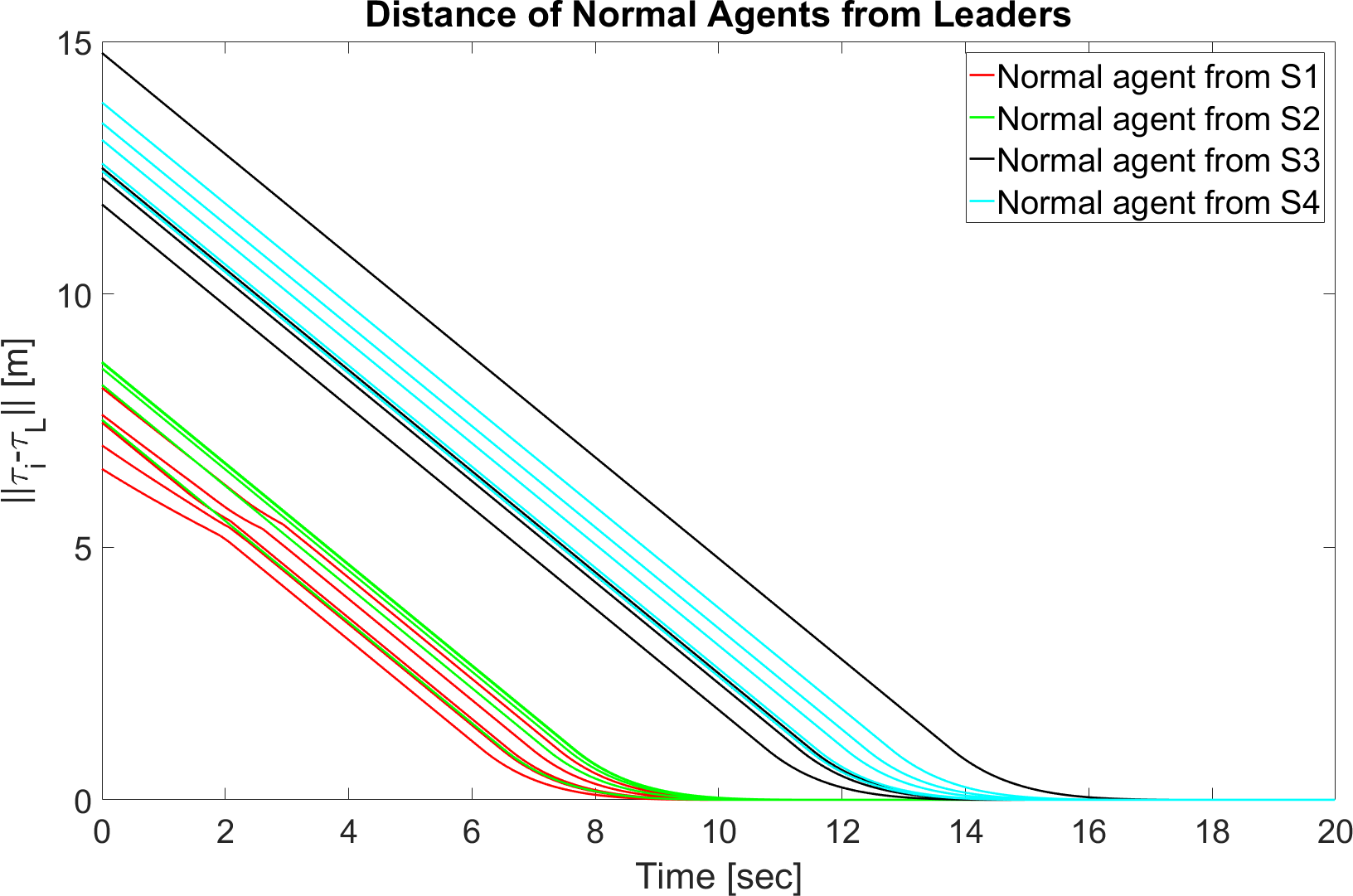}
	\caption{Norm $\|\bm \tau_i(t)-\bm \tau_L\|$ of a subset of the normal agents in the continuous time case. For sake of clarity, only a few normal nodes from each set $\Sc_p$ are shown.}
	\label{fig:norm cont}
\end{figure}

\begin{figure}[!htbp]
	\centering
	\includegraphics[width=0.9\columnwidth,clip]{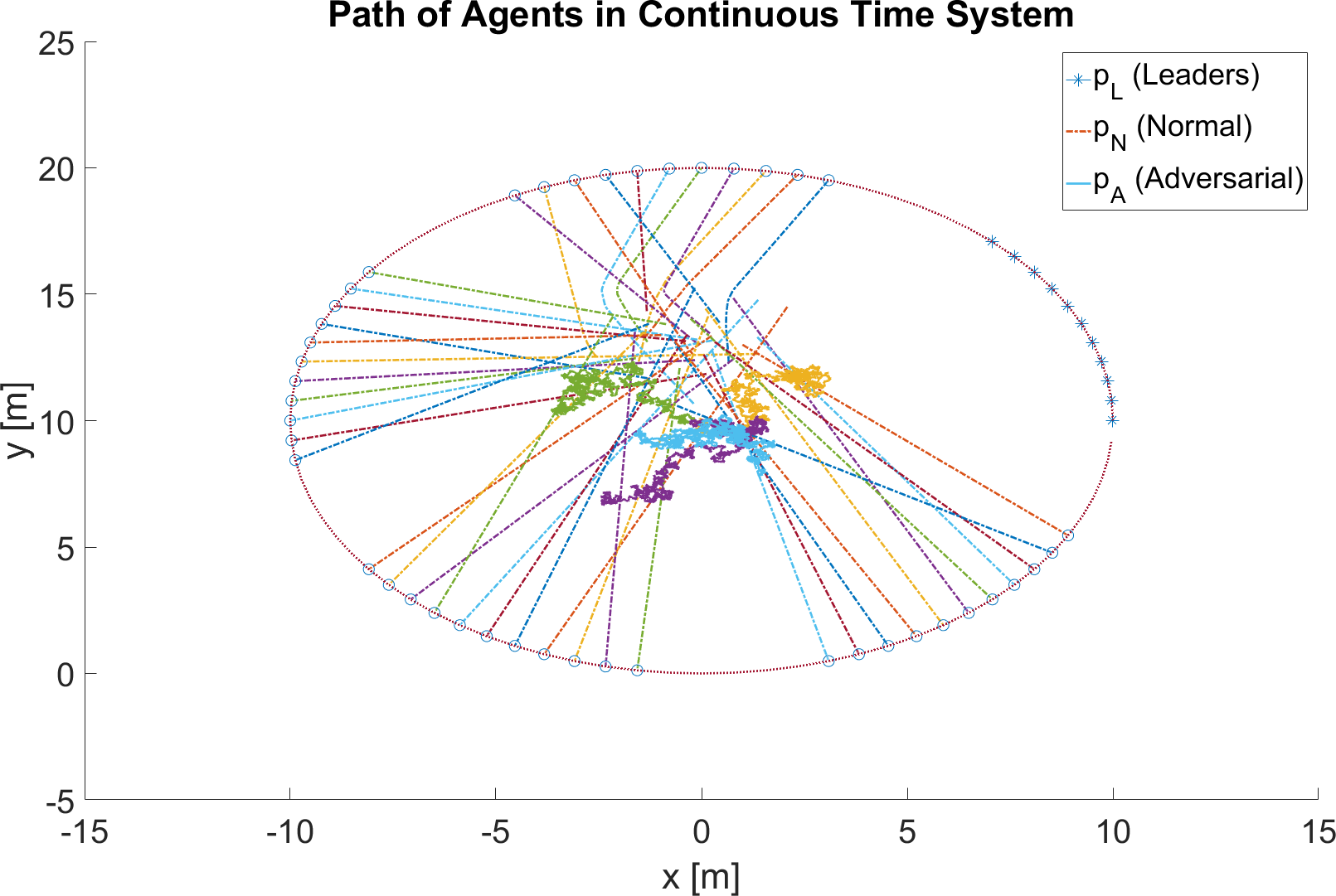}
	\caption{Path of the agents in the continuous time case. All normal and misbehaving agents start from the centre of the circle marked by red dots. The leaders are denoted by the star points $\textrm p_{\textrm L}$,
non-adversarial agents are denoted by $\textrm p_{\textrm N}$ and the adversarial agents are denoted as $\textrm p_{\textrm A}$.
}
	\label{fig:path cont}
\end{figure}

\begin{figure}[!htbp]
	\centering
	\includegraphics[width=0.9\columnwidth,clip]{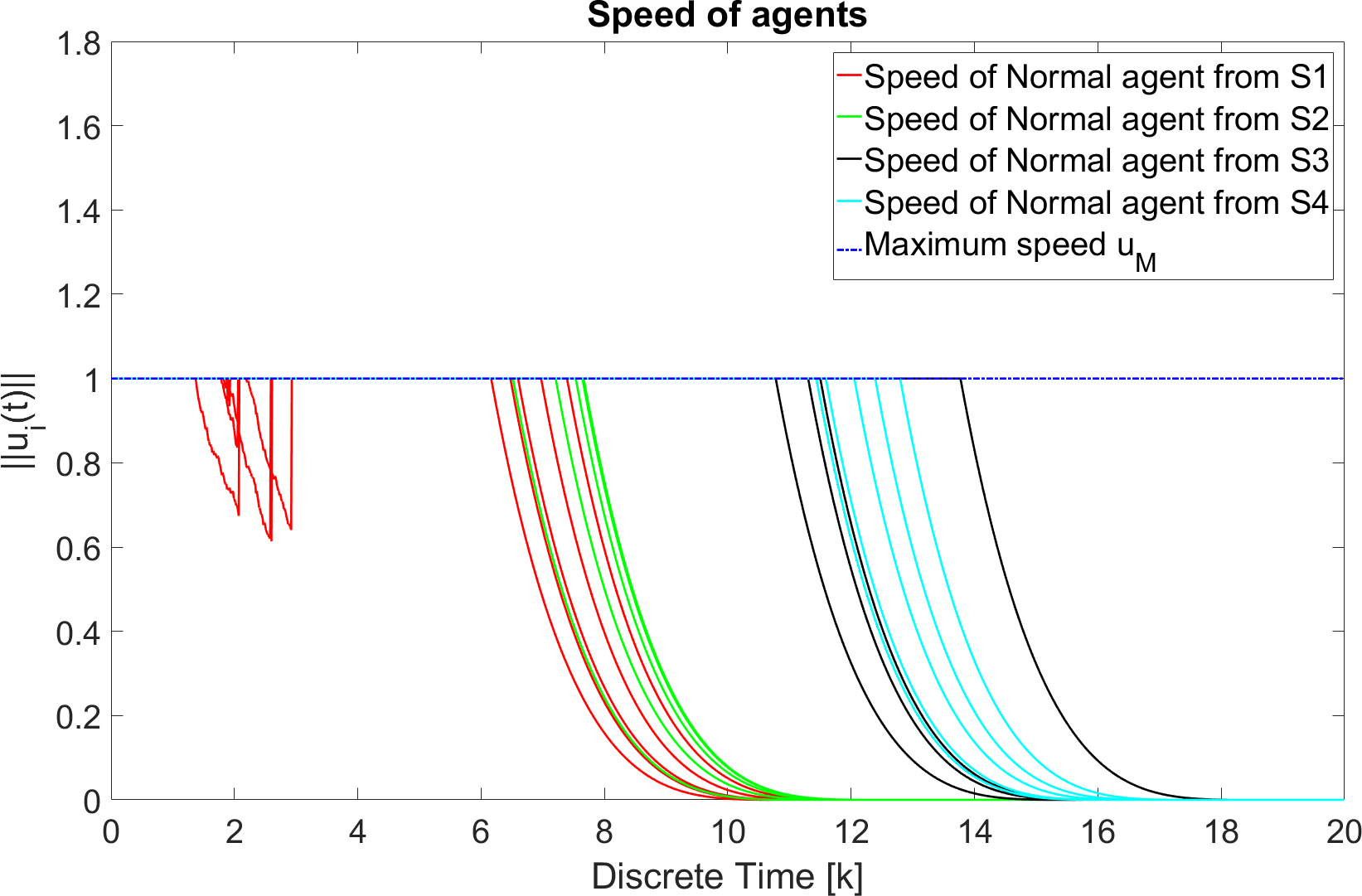}
	\caption{Norm $\|\bm u_i(t)\|$ of a subset of the normal agents, demonstrating that their input magnitudes never exceed the bound $u_M =1$. The rest of the network is not shown for sake of clarity.
	}
	\label{fig:speed cont}
\end{figure}

For the case of discrete system, Figure \ref{fig:norm disc} shows the variation of $\|\bm \tau_i[k]-\bm \tau_L\|$ with number of steps. Figure \ref{fig:path disc} shows the path $\bm p_i[k] = \bmx{x_i[k] & y_i[k]}^T$ of the agents. From both the figures, it is clear that despite the 5-local adversarial model, each normal agent achieves the desired formation.

From Figures \ref{fig:norm cont} and \ref{fig:norm disc}, it can be seen that agents in $\Sc_l$ converge before agents in $\Sc_{l+1}$, which is consistent with our analysis for this particular worst-case scenario.

\begin{figure}[!htbp]
	\centering
	\includegraphics[width=0.9\columnwidth,clip]{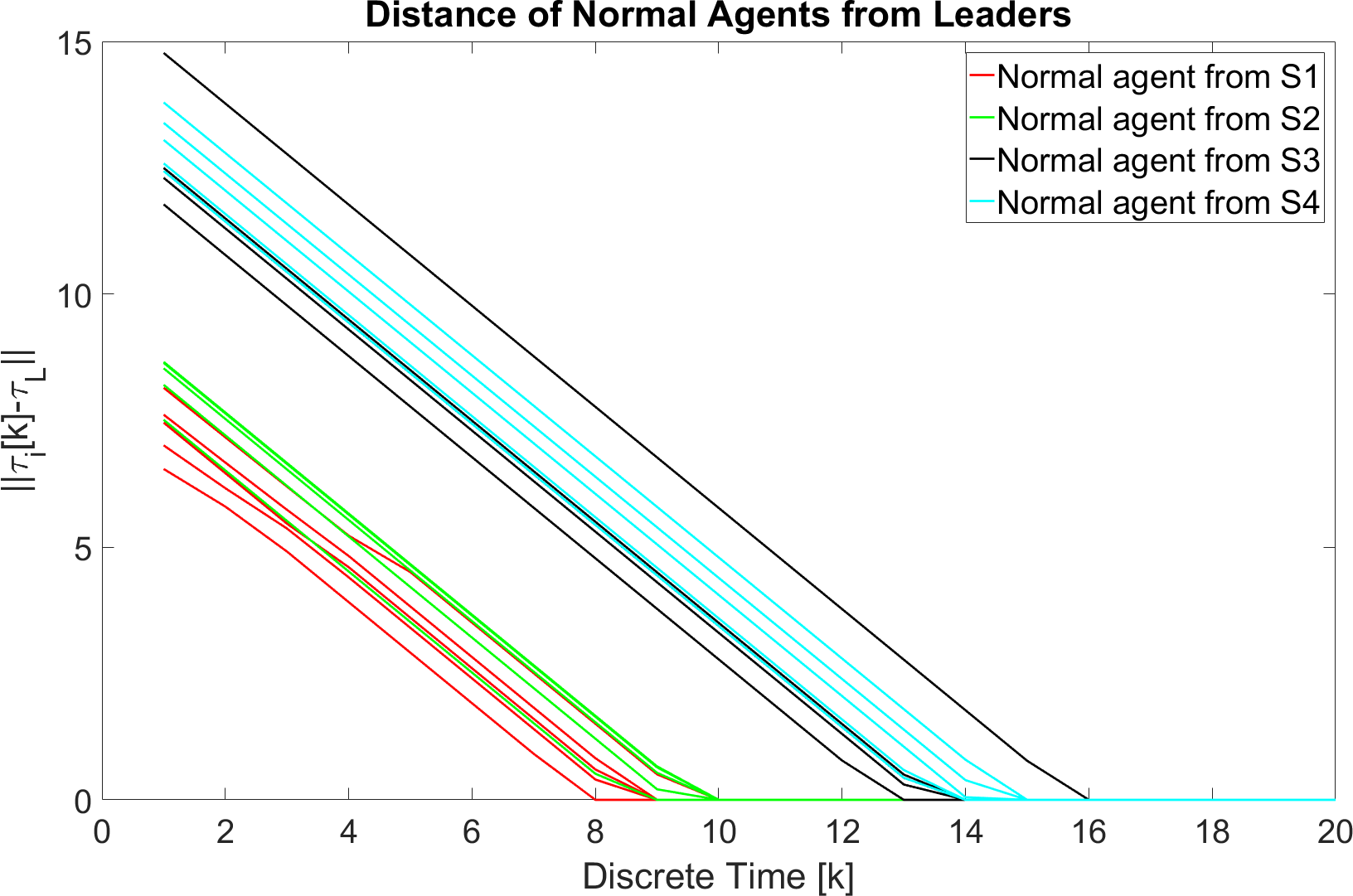}
	\caption{Norm $\|\bm \tau_i[k]-\bm \tau_L\|$ of a subset of the normal agents in the continuous time case. For sake of clarity, only a few normal nodes from each set $\Sc_p$ are shown.}
	\label{fig:norm disc}
\end{figure}

\begin{figure}[!htbp]
	\centering
	\includegraphics[width=0.9\columnwidth,clip]{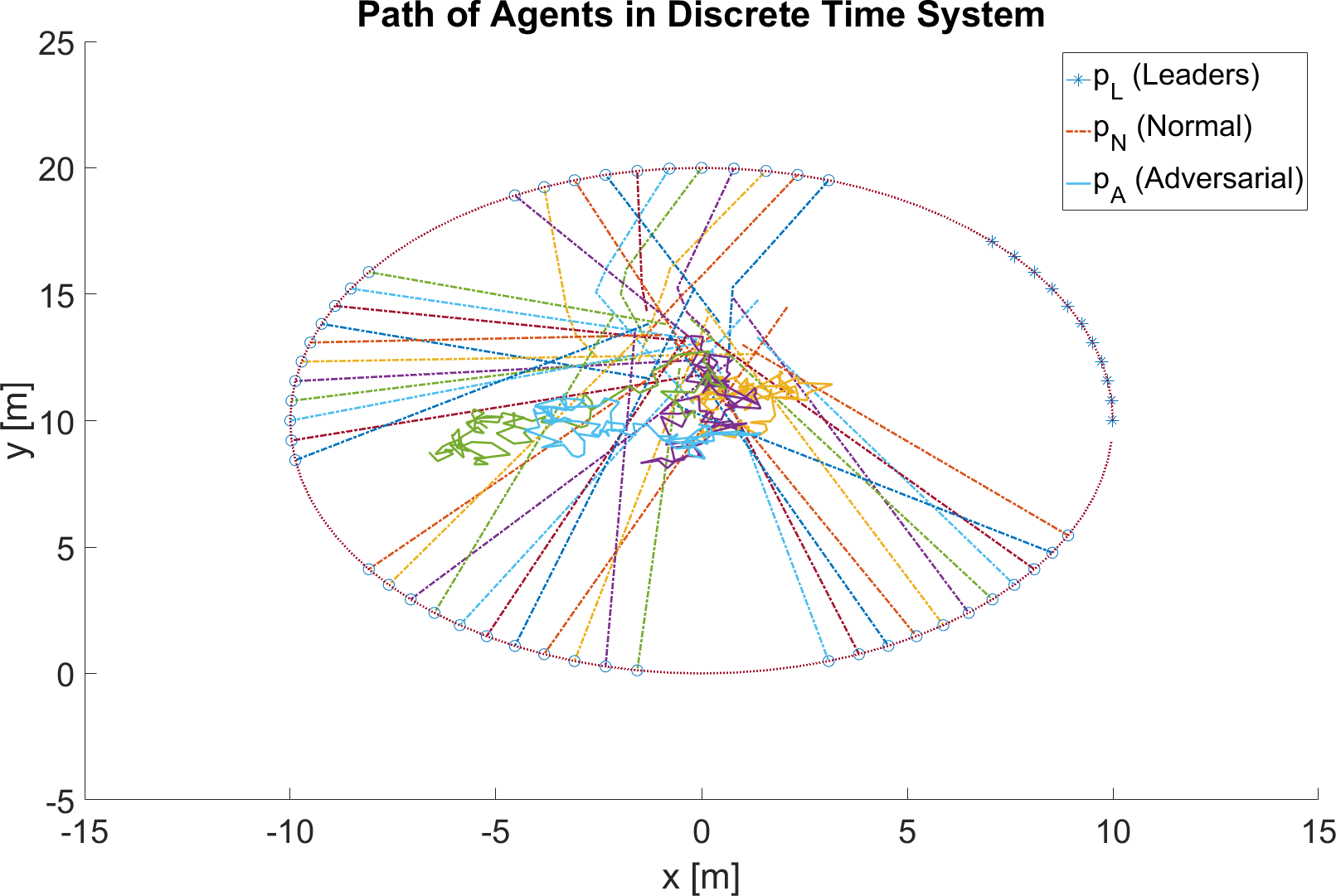}
	\caption{Path of the agents in the discrete time case. The leaders are denoted by the star points $\textrm p_{\textrm L}$,
non-adversarial agents are denoted by $\textrm p_{\textrm N}$ and the adversarial agents are denoted as $\textrm p_{\textrm A}$.
}
	\label{fig:path disc}
\end{figure}

\begin{figure}[!htbp]
	\centering
	\includegraphics[width=0.9\columnwidth,clip]{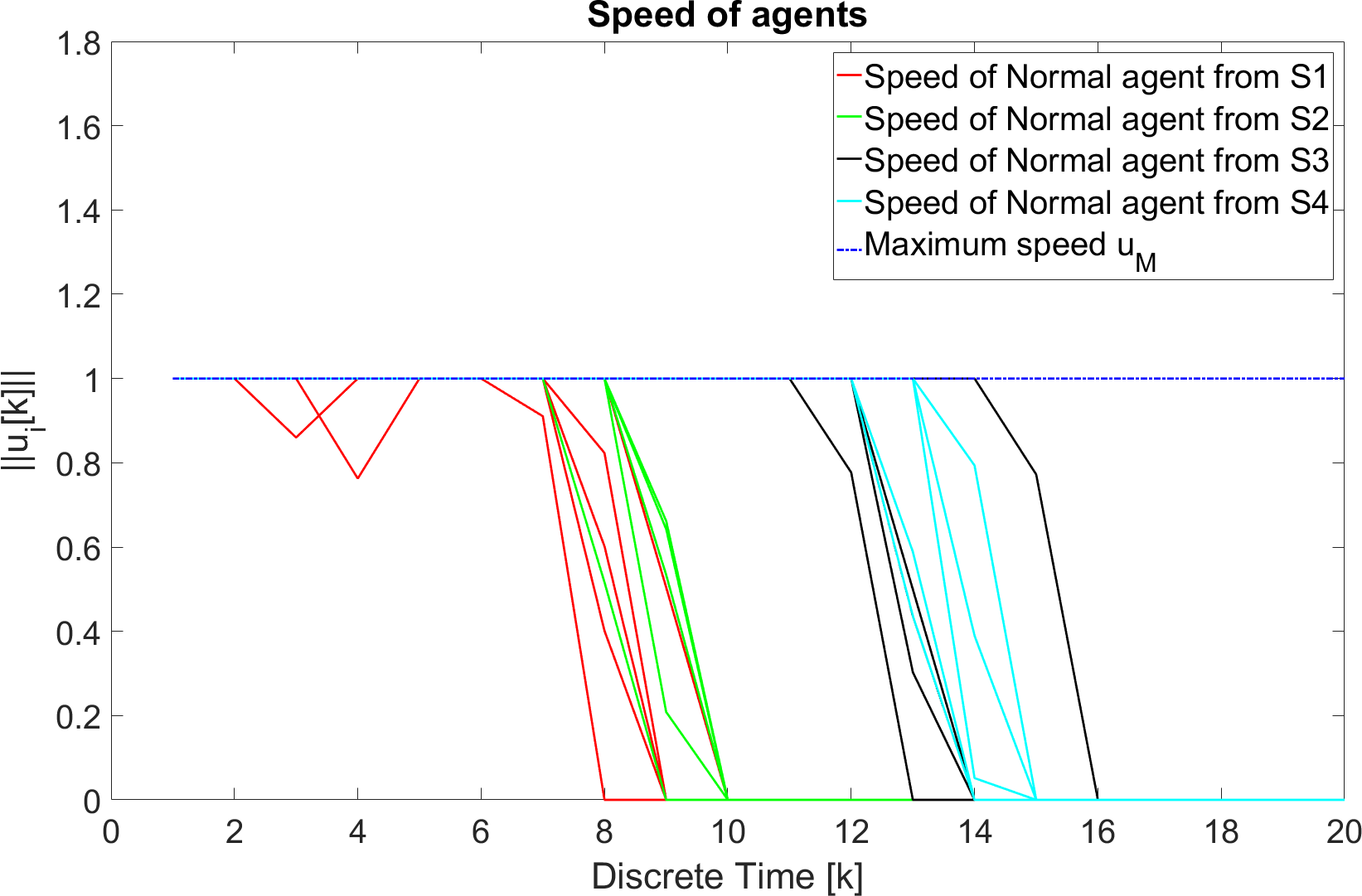}
	\caption{Norm $\|\bm u_i[k]\|$ of a subset of the normal agents in the discrete time case. 
 Again, the magnitude of each agents' control input never exceeds the bound $u_M = 1$ and goes to zero as the agents converge to formation.
	}
	\label{fig:speed disc}
\end{figure}

\section{Conclusion}
\label{sec:conc}

In this paper we introduced a novel continuous time resilient controller which guarantees that normally behaving agents can converge to a formation with respect to a set of leaders in the presence of adversarial agents. We proved that even with bounded inputs, the controller guarantees convergence in finite-time. In addition, we also applied our filtering mechanism to a discrete-time system and showed that it guarantees exponential convergence of agents to formation in the presence of adversaries under bounded inputs. Future work in this area will include further analysis of establishing safety among the normal agents and extending our results to time-varying graphs.



\bibliographystyle{IEEEtran}

\bibliography{Mendeley.bib}


\end{document}